\documentclass[a4paper,11pt]{article}

\usepackage[utf8]{inputenc}
\usepackage[T1]{fontenc}
\usepackage[english]{babel}

\usepackage{microtype}

\usepackage[pdftex]{graphicx}
\usepackage[usenames,dvipsnames]{xcolor}

\usepackage{amsmath,amssymb,amsfonts,amsthm}
\usepackage{mathabx}

\usepackage{fullpage}

\usepackage{booktabs}
\usepackage{microtype}

\newcommand{\eps}{\varepsilon}
\newcommand{\Ot}{\tilde{O}}

\newtheorem{claim}{Claim}
\newtheorem{conjecture}{Conjecture}
\newtheorem{observation}{Observation}
\newtheorem{lemma}{Lemma}
\newtheorem{theorem}{Theorem}

\usepackage{authblk}

\title{Tight Hardness Results for Distance and Centrality Problems in Constant
Degree Graphs}

\author[1]{Søren Dahlgaard\thanks{Research partly supported by Mikkel Thorup's
advanced grant DFF-0602-02499B from the Danish Council for Independent
Research.}}
\author[1]{Jacob Evald}
\affil[1]{University of Copenhagen \\\texttt{[soerend,jeav]@di.ku.dk}}

\begin{document}
\maketitle

\begin{abstract}
    Finding important nodes in a graph and measuring their importance is a
    fundamental problem in the analysis of social networks, transportation
    networks, biological systems, etc. Among the most popular such metrics of
    importance are graph centrality, betweenness centrality (BC), and reach
    centrality (RC). These measures are also very related to classic notions like
    diameter and radius. Roditty and Vassilevska Williams~[STOC'13] showed
    that no algorithm can compute a $(3/2-\delta)$-approximation of the
    diameter in sparse and unweighted graphs faster that $n^{2-o(1)}$ time
    unless the widely believed strong exponential time hypothesis (SETH) is
    false. Abboud et al.~[SODA'15] and [SODA'16] further analyzed these
    problems under the recent and very active line of research on hardness in
    \textbf{P}. They showed that in sparse and unweighted graphs (weighted for
    BC) none of these problems can be solved faster than $n^{2-o(1)}$
    unless some popular conjecture is false. Furthermore they ruled out a
    $(2-\delta)$-approximation for RC, a $(3/2-\delta)$-approximation for
    Radius and a $(5/3-\delta)$-approximation for computing all eccentricities
    of a graph for any $\delta > 0$.

    In this paper we extend these results to the case of unweighted graphs with
    constant maximum degree. Through new graph constructions we are able to
    obtain the same approximation and time bounds as for sparse graphs even in
    unweighted graphs with maximum degree 3. Specifically we show that no
    $(3/2-\delta)$ approximation of Radius or Diameter,
    $(2-\delta)$-approximation of RC, $(5/3-\delta)$-approximation of all
    eccentricities or exact algorithm for BC exists in time $n^{2-o(1)}$ for
    such graphs and any $\delta > 0$. For BC, this strengthens the result of
    Abboud et al.~[SODA'16] by showing a hardness result for unweighted graphs.
    Our results follow in the footsteps of Abboud et al.~[SODA'16] and Abboud
    and Dahlgaard~[FOCS'16] by showing conditional lower bounds for restricted
    but realistic graph classes.
\end{abstract}

\section{Introduction}
Measuring the importance of specific nodes in a graph is a fundamental
problem in the analysis of social networks, transportation networks, biological
systems, etc. Several notions of importance have been proposed in the
literature. Among the most popular such notions are several centrality measures
such as \emph{graph centrality}~\cite{Hage95}, \emph{betweenness
centrality} (BC)~\cite{Freeman77}, \emph{closeness centrality}~\cite{Sabidussi66},
and \emph{reach centrality}~\cite{Gutman04}. All these centrality measures are
closely related to the shortest paths of the graph. As an example, the graph
centrality of a node is the inverse of its maximum distance to any other node
in the graph. These measures are extensively studied in both the theoretical
and practical communities with some papers having thousands of
citations~\cite{Brandes01,Freeman77,Hakimi64,Newman05}. Fully understanding the
complexity of computing these measures is thus a very important problem in the
field of graph analysis. In this paper we follow a recent and very active line
of research on showing hardness or \emph{conditional lower bounds} (CLBs) for
problems in
\textbf{P}~\cite{Patrascu10,AbboudV14,AbboudVY15,AbboudGV15,Dahlgaard16,AbboudWW16,AbboudD16,RodittyW13}
focusing on the above mentioned centrality measures in graphs with constant
degrees.

Given an undirected and unweighted graph $G = (V,E)$ with $n$ nodes and $m$
edges we let $d_G(u,v)$ denote the distance between nodes $u,v\in V$. We omit
the subscript $G$, when it is clear from the context. Let $\sigma_{s,t}$
denote the number of distinct shortest paths between $s$ and $t$ and let
$\sigma_{s,t}(u)$ denote the number of such paths passing through $u$. In this
paper we consider several centrality and importance measures of graphs. We
summarize the definitions of these below in Table~\ref{tab:defns}. We note that
all of these definitions also make sense for weighted graphs, but we
concentrate on unweighted graphs in this paper, as they are more difficult to
show hardness results for in general (see e.g.~\cite{AbboudD16}).
\begin{table}[htbp]
    \centering
    \begin{tabular}{l | c}
        \toprule
        \textbf{Name} & \textbf{Definition} \\
        \midrule
        The \emph{Eccentricity} (of $u$) & $\epsilon(u) := \max_{v\in V} d(u,v)$ \\
        The \emph{Diameter} & $D := \max_{u\in V} \epsilon(u)$ \\
        The \emph{Radius} & $R := \min_{u\in V} \epsilon(u)$ \\
        The \emph{Reach Centrality (RC)} (of $u$) & $RC(u) :=
        \max_{\substack{s,t\in V\\d(s,t) = d(s,u)+d(u,t)}}\left(
        \min\!\left(d(s,u),d(u,t)\right)\right)$ \\
        The \emph{Betweenness Centrality (BC)} (of $u$) & $BC(u) :=
        \sum_{s,t\in V\setminus \{u\}}\frac{\sigma_{s,t}(u)}{\sigma_{s,t}}$ \\
        The \emph{Graph Centrality (GC)} (of $u$) & $GC(u) :=
        \frac{1}{\max_{v\in V} d(u,v)}$ \\
        \bottomrule
    \end{tabular}
    \caption{Definitions of different centrality and importance measures.}
    \label{tab:defns}
\end{table}
We note that the maximum Graph Centrality is exactly the inverse of the Radius.
All of these measures except BC can be computed by simply running an algorithm
for the classical all pairs shortest paths (APSP) problem in $\Ot(n^\omega)$ or
$\Ot(mn)$ time, where $\omega$ is the matrix-multiplication exponent. For BC we
can use Brandes's algorithm~\cite{Brandes01} to compute the betweenness
centrality of all nodes in $O(mn + n^2\log n)$ time\footnote{We note that
Brandes's and other popular algorithms for computing BC neglects the complexity
of keeping the counters of the number of shortest paths. If taking this into
account the worst-case running time grows by a factor of $\Theta(n\log n)$.}.

\subsection{Hardness in P}\label{sec:hardness}
A recent and very active line of work concerns itself with showing hardness
results for problems in \textbf{P} based on the assumption of several popular
conjectures. For the measures in Table~\ref{tab:defns} several results are
known. Perhaps the most well-studied of the problems from a theoretical
perspective is diameter. Roditty and Vasillevska Williams~\cite{RodittyW13}
showed that no algorithm can solve the diameter problem in sparse graphs in
time $O(n^{2-\eps})$ for any $\eps>0$ unless the widely believed Strong
Exponential Time Hypothesis (SETH)~\cite{ImpagliazzoP01} is false. In fact they
showed that no algorithm can even compute at $3/2-\delta$ approximation in
this time for any $\delta > 0$. We say that a number $x$ is an
$\alpha$-approximation of a number $y$ if $y \le x\le \alpha y$. In some cases
we will also allow the algorithm to provide an under-approximation.
Abboud et al.~\cite{AbboudWW16} showed similar results for the problem
of computing the radius, median, and all eccentricities in sparse graphs.
They showed that unless SETH is false no algorithm can compute a
$5/3-\delta$ approximation of all eccentricities in time $O(n^{2-\eps})$
for any $\eps,\delta > 0$ in sparse graphs. Based on the similar \emph{Hitting
Set (HS) Conjecture} they also showed hardness results for the radius and
median problems. For radius, they showed that no algorithm can compute a
$3/2-\delta$ approximation in time $O(n^{2-\eps})$ for any $\eps,\delta>0$.

For the centrality measures of Table~\ref{tab:defns}, Abboud et
al.~\cite{AbboudGV15} showed that radius, median, and betweenness centrality
are all equivalent to the classic APSP problem under \emph{subcubic
reductions}. Similarly, they showed that RC is equivalent to the diameter
problem under subcubic reductions. For betweenness centrality they showed
that computing an $\alpha$-approximation for any $\alpha > 0$ is equivalent to
APSP under subcubic reductions and that no algorithm can compute such an
approximation for any node in time $O(n^{2-\eps})$ in sparse graphs unless SETH
is false. Finally they show that computing reach centrality is equivalent to
diameter under subcubic reductions and that computing a
$(2-\delta)$-approximation of RC in sparse and \emph{unweighted} graphs cannot
be done in $O(n^{2-\eps})$ time unless SETH is false. An important note about
all these reductions except for RC in sparse graphs is that they only hold for
\emph{weighted} graphs.

The known hardness results for the measures in Table~\ref{tab:defns} are
summarized below in Table~\ref{tab:results}.

\begin{table}[htbp]
    \centering
    \begin{tabular}{l | l | l | l | c}
        \toprule
        \textbf{Problem} & \textbf{Bound} & \textbf{Approximation} &
        \textbf{Graph family} & \textbf{Source} \\
        \midrule
        Diameter & $n^{2-o(1)}$ & $3/2-\delta$-approx & Sparse, unweighted &
        \cite{RodittyW13} \\
        Radius & $n^{2-o(1)}$ & $3/2-\delta$-approx & Sparse,
        unweighted & \cite{AbboudWW16} \\
        Radius & $n^{3-o(1)}$ & Exact & Dense, weighted & \cite{AbboudGV15} \\
        Eccentricities & $n^{2-o(1)}$ & $5/3-\delta$ & Sparse, unweighted & \cite{AbboudWW16} \\
        Reach Centrality & $n^{2-o(1)}$ & $2-\delta$ & Sparse, unweighted &
        \cite{AbboudGV15} \\
        Reach Centrality & $n^{2-o(1)}$ & Any finite & Sparse, weighted &
        \cite{AbboudGV15} \\
        Betweenness Centrality & $n^{2-o(1)}$ & Any finite & Sparse, weighted &
        \cite{AbboudGV15} \\
        Betweenness Centrality & $n^{3-o(1)}$ & Exact & Dense, weighted &
        \cite{AbboudGV15} \\
        \bottomrule
    \end{tabular}
    \caption{Known hardness results for the measures of Table~\ref{tab:defns}.
    The hardness results are based on different conjectures. We refer to the
    discussion above for the specific details.}
    \label{tab:results}
\end{table}

In this paper we will follow this active line of research on hardness in
\textbf{P}, basing hardness on the two following popular conjectures.
\begin{conjecture}[Orthogonal Vectors Conjecture]\label{conj:ov}
    Let $A,B\subseteq \{0,1\}^d$ be two sets of $n$ boolean vectors of
    dimension $d = \omega(\log n)$. Then there exists no algorithm that can
    determine whether there is an orthogonal pair $a\in A$, $b\in B$ in time
    $O(n^{2-\eps})$ for any $\eps > 0$.
\end{conjecture}
The OV conjectures is implied by SETH~\cite{Williams05} and has been used in
several papers as an intermediate step for showing SETH-hardness.

\begin{conjecture}[Hitting Set Conjecture]\label{conj:hs}
    Let $A,B\subseteq \{0,1\}^d$ be two sets of $n$ boolean vectors of
    dimension $d = \omega(\log n)$. Then there exists no algorithm that can
    determine whether there is an $a\in A$ that is not orthogonal to any $b\in
    B$ in time $O(n^{2-\eps})$ for any $\eps > 0$.
\end{conjecture}
We say that a vector $a\in A$ that is not orthogonal to any $b\in B$ is a
\emph{hitting set}, as it ``hits'' each vector of $B$. The HS conjecture was
introduced by Abboud et al.~\cite{AbboudWW16} as an analog to the OV
conjecture. One reason for introducing this conjecture is that Carmosino et
al.~\cite{CarmosinoGIMPS16} give evidence that problems such as radius and
median cannot be shown to have SETH-hardness due to the nature of the
quantifiers in this problem.

\subsection{Hardness for graphs with constant max degree}
An important direction for future research on fine-grained complexity of
polynomial time problems, as noted by Abboud et al.~\cite{AbboudWW16} and
Abboud and Dahlgaard in~\cite{AbboudD16} is to understand the complexity of
fundamental problems in restricted, but realistic classes of graphs. Such
results for graphs of bounded treewidth was shown by Abboud et
al.~\cite{AbboudWW16} and for planar graphs by Abboud and
Dahlgaard~\cite{AbboudD16}. The result of~\cite{AbboudD16} even holds for
planar graphs with constant maximum degree. Following up on this work, we
concentrate on showing hardness for problems in graphs with constant maximum
degree (bounded degree graphs). This family of graphs has received attention in
several communities and in some cases exhibits better algorithms and data
structures than graphs with unbounded degrees. This includes approximation of
NP hard problems~\cite{Halperin02,AustrinKS11,Halldorsson00}, labeling
schemes~\cite{AlonN16,AlstrupDKP16}, property testing~\cite{GoldreichR02} and
graph spanners~\cite{FominGL11}. Furthermore, ruling out fast algorithms for
graphs with constant maximum degree also rules out approaches based on bounded
arboricity or degeneracy. Such approaches have previously been used to
significantly speed up algorithms for e.g. subgraph finding (see
e.g.~\cite{ChibaN85,AlonG09,AlonYZ97}. See also \cite{DahlgaardKS17} which
improved on \cite{AlonYZ97} by showing a structural lemma for graphs with low
degree).

Extending the results of conditional lower
bounds from~\cite{AbboudGV15,AbboudWW16,RodittyW13} to bounded degree graphs
thus seems like a natural step in developing the theory of hardness in
\textbf{P}.

\subsection{Our results}
In this paper we present tight hardness results for all of the measures of
Table~\ref{tab:defns} in graphs with constant maximum degree.
More precisely we show the following theorems.

\begin{theorem}\label{thm:diameter}
    No algorithm can compute a $(3/2-\delta)$-approximation of the diameter of
    an unweighted constant degree graph in time $O(n^{2-\eps})$ for any
    $\eps,\delta > 0$ unless Conjecture~\ref{conj:ov} is false.
\end{theorem}
\begin{theorem}\label{thm:radius}
    No algorithm can compute a $(3/2-\delta)$-approximation of the radius of
    an unweighted constant degree graph in time $O(n^{2-\eps})$ for any
    $\eps,\delta > 0$ unless Conjecture~\ref{conj:hs} is false.
\end{theorem}
\begin{theorem}\label{thm:eccentricities}
    No algorithm can compute a $(5/3-\delta)$-approximation of all
    eccentricities in an unweighted constant degree graph in time
    $O(n^{2-\eps})$ for any $\eps,\delta > 0$ unless Conjecture~\ref{conj:ov}
    is false.
\end{theorem}
\begin{theorem}\label{thm:rc}
    No algorithm can compute a $(2-\delta)$-approximation of the reach
    centrality of any node in an unweighted constant degree graph in time
    $O(n^{2-\eps})$ for any $\eps,\delta > 0$ unless Conjecture~\ref{conj:ov}
    is false.
\end{theorem}
\begin{theorem}\label{thm:bc}
    No algorithm can compute the betweenness centrality of any node in an
    unweighted constant degree graph in time $O(n^{2-\eps})$ for any $\eps > 0$
    unless Conjecture~\ref{conj:ov} is false.
\end{theorem}
All of our results above match the corresponding hardness results for sparse
graphs, extending them to graphs with constant maximum degree. Furthermore, our
hardness result for betweenness centrality improves on the result
of~\cite{AbboudGV15} in general graphs by ruling out truly subquadratic
algorithms in \emph{unweighted} graphs.
As mentioned all of these problems can be solved in $\Ot(n^2)$ or even $O(n^2)$
time in unweighted graphs with constant max degree\footnote{If we assume unique
or even polynomial number of shortest paths for BC.}, thus all of our results
are tight. Furthermore, a more efficient $\Ot(m\sqrt{n})$ algorithm exists for
$3/2$-approximating the diameter of a graph~\cite{RodittyW13,ChechikLRSTW14}
and for $5/3$-approximating all eccentricities of the
graph~\cite{ChechikLRSTW14}. Thus, our approximation guarantees are also
optimal for these kinds of problems. We note that our reductions in general
produces CLBs for graphs with maximum degree 4, however a simple splitting of
nodes of degree 4 improves this to a maximum degree of 3.

\subsection{Challenges and techniques}
Many of the previous results in Table~\ref{tab:results} work for sparse graphs,
and it is often assumed that such reductions translate directly to bounded
degree graphs. A common technique in obtaining such reductions from sparse
graphs to constant degree graphs (see e.g.~\cite{AlstrupDKP16}) is to split
each vertex and insert edges of weight zero.
However, for the results of this paper we seek CLBs for \emph{unweighted}
graphs, and thus introducing edges of weight $0$ is not an option.
In~\cite{AbboudCK16} a near-linear lower bound is
given for computing a $(3/2-\delta)$-approximation of diameter in the CONGEST
model even in sparse networks. They also present a reduction to bounded-degree
graphs, however this reduction does not preserve the approximation ratio of the
lower bound and only excludes exact computation in bounded degree graphs. This
problem seems inherent in reductions from sparse graphs to graphs with bounded
maximum degree. The most common reductions replace nodes of degree $x$ by a
balanced binary tree with $x$ leaves, and each such leaf then has an edge to
the binary trees of the corresponding neighbours. By reducing in this way two
nodes in the new graph corresponding to the same node in the original graph may
have $\Omega(\log n)$ distance, while two nodes corresponding to different
original nodes may have distance $1$. Such an imbalance can be addressed by
replacing edges between trees by sufficiently long paths. However, this may
lead to worse approximation guarantees for problems like diameter,
since the longest distance in the graph may now involve a node ``in the
middle'' of such a path, which corresponds to an original edge rather than an
original node. As an example, applying this method to the diameter reduction of
Roditty and Vassilevska Williams~\cite{RodittyW13} would only rule out a
$(6/5-\delta)$-approximation and not the original factor of $3/2-\delta$. In
the reductions in this paper we therefore introduce reductions which are based
on similar ideas as \cite{RodittyW13,AbboudWW16,AbboudGV15}, but are carefully
tailored to graphs with constant degree in order to obtain the same
approximation guarantees. Furthermore, when reducing to betweenness centrality
we introduce many new nodes which we have to sum over. We avoid this
complicated counting argument by creating several related graphs and look at
the change in betweenness centrality to complete the reduction.

\section{Hardness of eccentricity-related problems}
In this section we will prove Theorems~\ref{thm:diameter},~\ref{thm:radius},
and~\ref{thm:eccentricities}. Our reductions are inspired by standard techniques
as used in~\cite{RodittyW13,AbboudWW16,Dahlgaard16,CairoGR16} and many others.
These reductions all use the OV-graph described below as a basis and add
different shortcuts and nodes based on the problem at hand. In the case of
constant degree graphs we have to be careful with the way we define deroutes,
as we seek to get the same approximation guarantees as in the sparse case (see
Table~\ref{tab:results}). A natural approach is to simply take the graph from
eg.~\cite{RodittyW13} and replace each node by a binary tree and each edge by a
sufficiently long path (say length $x$). In~\cite{RodittyW13} they create a
graph where the diameter is either $3$ or $2$, and by doing this reduction we
get a graph, where we would hope that the diameter is either $(3+o(1))x$ or
$(2+o(1))x$. However, since the edges are replaced by paths we instead end up with
either $(3+o(1))x$ or $(\frac{5}{2}+o(1))x$ giving a much worse bound on the
approximation. Below we show how to circumvent this bound by extending the
OV-graph more carefully.

We will now define a general graph construction, the $OV$-graph, for reductions
based on Conjecture \ref{conj:ov}, which has been used in several papers
previously~\cite{AbboudV14,AbboudWW16,Dahlgaard16,CairoGR16,AbboudGV15,Williams05}:

Given an instance $A,B\subseteq \{0,1\}^d$ of the OV problem, we construct the
\emph{$OV$-graph} as follows:
\begin{itemize}
    \item For each vector $a \in A$ and $b \in B$ create a node. We will
        overload the notation and use $a$ and $b$ to refer to both the vector
        and node when this is clear from the context.
    \item Create a set $C$ of $d$ nodes denoted $c_1, c_2, \ldots c_d$.
    \item For each node $a$ create an edge $(a, c_i)$ if $a[i] = 1$ for each $1
        \leq i \leq d$. Do the same for each node $b$.
\end{itemize}
The following observation is key in all reductions using this graph.
\begin{observation}\label{obs:ovgraph}
A pair $a \in A$, $b \in B$ is orthogonal if and only if $d(a, b) > 2$.
\end{observation}

For reductions based on Conjecture \ref{conj:hs} we will also use the
$OV$-graph, but since we are looking for an $a \in A$ that is \emph{not}
orthogonal to \emph{any} $b \in B$, we will instead try to determine whether
$d(a, b) = 2$ for all $b \in B$.

We note, that we will assume that each node $a$ and $b$ has an edge to at least
one $c_i$. Otherwise this would be the all zeroes vector, and we can easily
answer the OV problem.

\subsection{Diameter in constant degree graphs}\label{sec:diameter}
We are now ready to prove Theorem~\ref{thm:diameter}. Our reduction to diameter
in constant degree graphs uses the $OV$-graph with the following modifications:

\begin{itemize}
    \item For each node $u \in A \cup B$ insert a balanced binary tree with $d$
        leaves rooted at $u$. Denote the trees \emph{vector trees} and the root
        vertices $a_r$ and $b_r$ respectively.

    \item For each node $c \in C$ insert two balanced binary trees with $n$
        leaves, both rooted at $c$. One for handling edges to $A$, and one for
        $B$. Denote the tree for handling edges to $A$ and $B$ respectively the
        $c_A$\emph{-tree} and the $c_B$\emph{-tree}.

    \item For each each edge $(a, c)$ of the $OV$-graph, insert a path of
        length $p$ (to be fixed later) from a leaf in the tree rooted
        at $a_r$, to a leaf in the $c_A$-tree. Do the same for edges
        $(b,c)$. We assign paths to leaves in such a way that each leaf has
        degree at most $2$ in the final graph.

    \item Create a balanced binary tree with $n$ leaves and connect each leaf
        to a node $a_r$ with a path of length $p$, such that each $a_r$ is
        connected to a distinct leaf. Do the same for the nodes $b_r$.
        Call these trees the \emph{$A$-shortcut tree} and \emph{$B$-shortcut
        tree} respectively.

    \item At each root $a_r$ and $b_r$ of a vector tree, add a path of length
        $p$ and denote the end of each path by $a_p$.
\end{itemize}
We call this the $OV_{dia}$-graph. Figure~\ref{fig:diameter} contains an
example $OV_{dia}$-graph for reference. It is easy to verify that this graph
has maximum degree $4$ with only the nodes $a_r$, $b_r$, and $c_i$ having this
degree.
\begin{figure}[ht]
    \centering
    \includegraphics[width=1.0\textwidth]{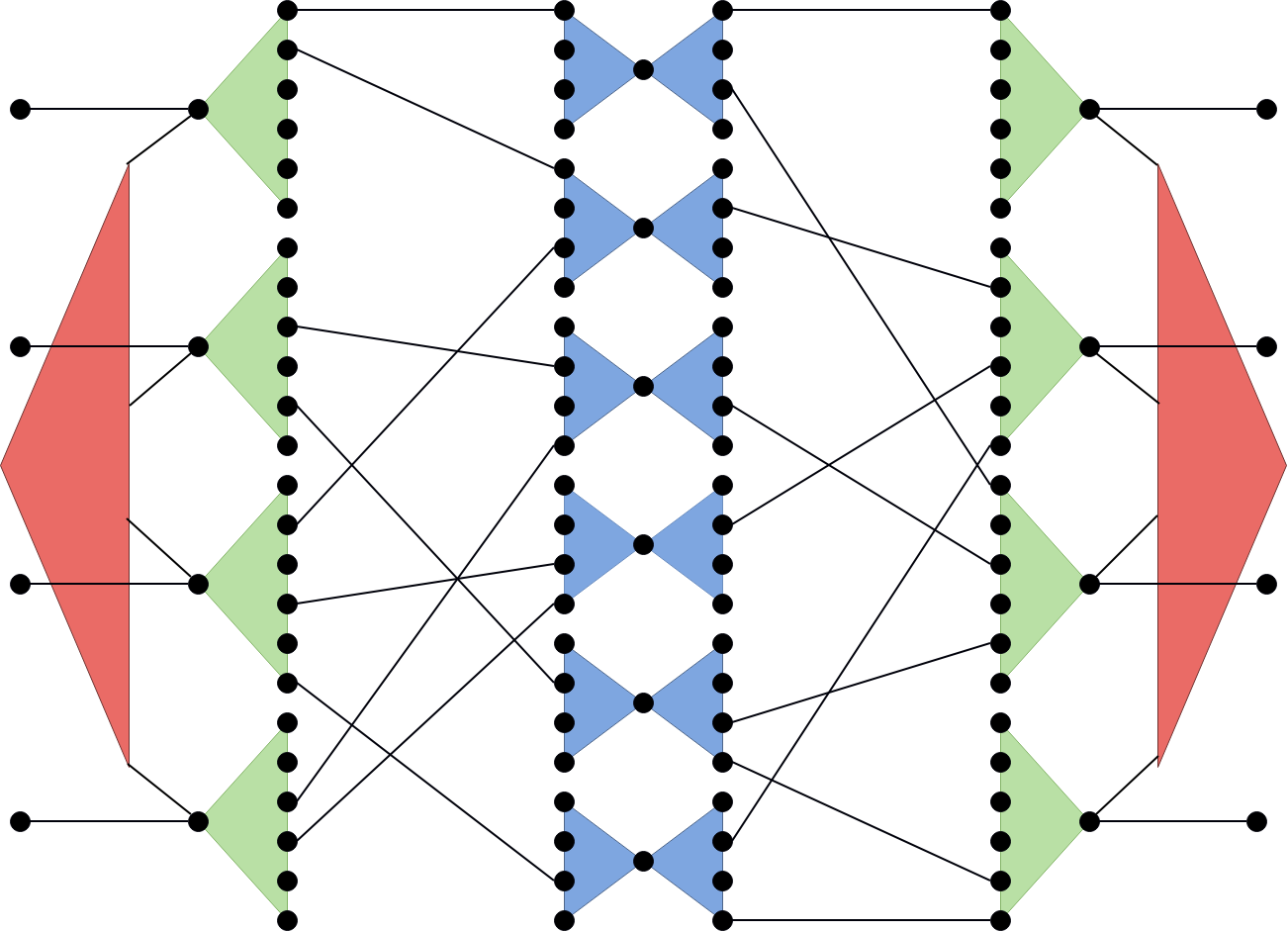}
    \caption{The $OV_{dia}$-graph, note that all lines in the graph are paths of length $p$. Green triangles are vector trees, blue triangles are $c_A$ and $c_B$ trees and red triangles are the $A$ and $B$-shortcut trees}
    \label{fig:diameter}
\end{figure}
We make the following claims about distances in the $OV_{dia}$-graph, the
proofs can be found in appendix \ref{appendix:diaproofs}:

\begin{claim}\label{claim:notorthogonal}
    For two vectors $a \in A$, $b \in B$ that are \emph{not} orthogonal,
    $d(a_p, b_p) \leq 4p + 2 \lg n + 2 \lg d$.
\end{claim}

\begin{claim}\label{claim:orthogonal}
    For two vectors $a \in A$, $b \in B$ that are orthogonal, $d(a_p, b_p) \geq
    6p$.
\end{claim}

\begin{claim}\label{claim:vectortrees}
    For two nodes $u$ and $v$ in vector trees of $A$, $d(u, v) \leq 2p + 2\lg n + 2\lg d$.
    Symetrically for $B$.
\end{claim}

\begin{claim}\label{claim:cnode}
    For a node $u$ in some $c_A$ or $c_B$ tree, $\epsilon(u) \leq 4p + 4 \lg n + 2
    \lg c$.
\end{claim}

\begin{claim}\label{claim:shortcut}
    For a node $u$ in the $A$-(or $B$-)shortcut tree, $\epsilon(u) \leq 4p + 4\lg n + 2
    \lg c$.
\end{claim}

The following lemma now follows directly from Claims~\ref{claim:notorthogonal}
through~\ref{claim:shortcut} above and by picking $p =\omega(\log n)$
sufficiently large.
\begin{lemma}
    The $OV_{dia}$-graph has diameter $(6+o(1))p$ if an orthogonal pair exists,
    and $(4+o(1))p$ if no orthogonal pair exists.
\end{lemma}
Using this Lemma we can now prove Theorem~\ref{thm:diameter}: Assume that there
exists an algorithm with running time $O(n^{2-\eps})$ for computing a
$(3/2-\delta)$-approximation of the diameter of a constant degree graph for
some $\eps,\delta > 0$. We then create the $OV_{dia}$-graph for an instance of
the OV problem, and run our algorithm. If no orthogonal pair exists the
algorithm will return at most $(6-\delta' + o(1))p \ll 6p$ for some fixed
$\delta' > 0$ and we can thus correctly decide whether an orthogonal vector
pair exists. Since the graph has $\Ot(n)$ nodes and edges and can be
constructed in $\Ot(n)$ time, this implies a $O(n^{2-\eps})$ algorithm for the
OV problem contradicting Conjecture~\ref{conj:ov}.

\subsection{Radius in constant degree graphs}
In order to prove Theorem~\ref{thm:radius} we will use a similar approach as in
the previous section, however we will be reducing from the HS problem of
Conjecture~\ref{conj:hs}. We will use the following graph, $OV_{rad}$:

Take two copies of the $OV_{dia}$-graph defined in the previous section, $G_1$
and $G_2$ and glue the graphs together in the nodes nodes $a_p$, such that
$G_2$ is a ``mirrored'' $G_1$ along the nodes $a_p$. This is illustrated in
figure \ref{fig:radius} below.
\begin{figure}[htbp]
    \centering
    \includegraphics[width=1.0\textwidth]{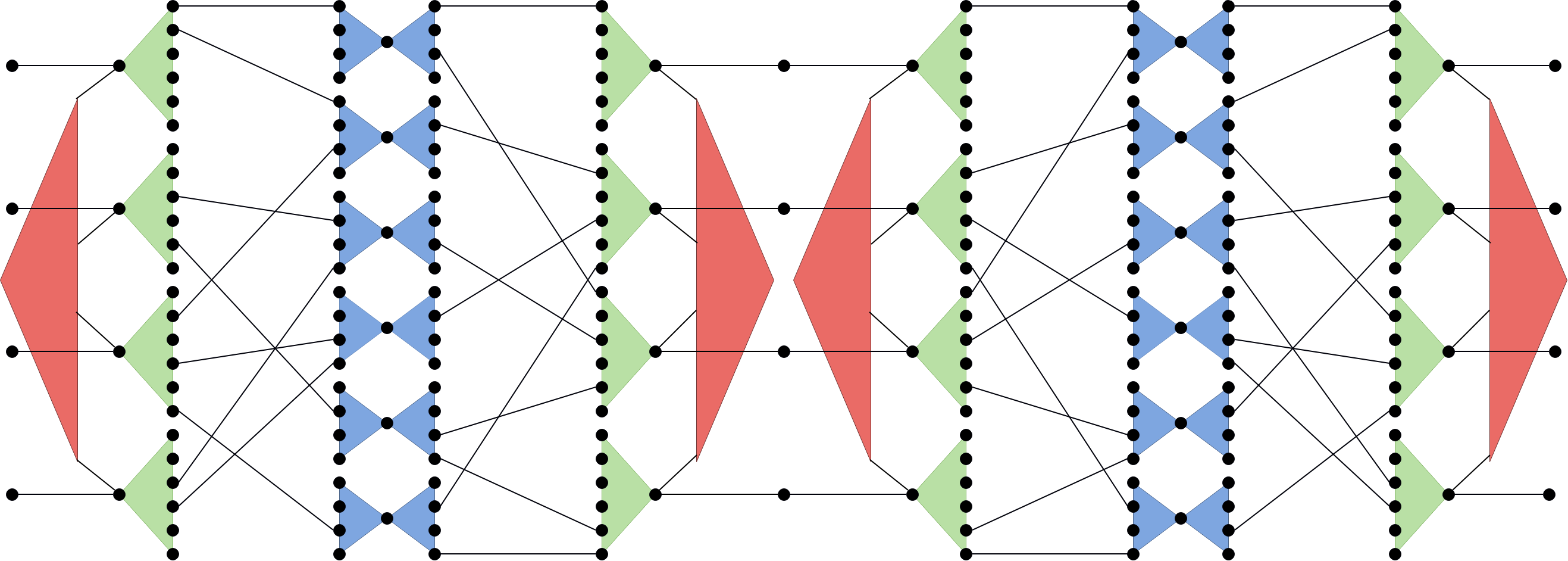}
    \caption{The $OV_{rad}$-graph, as in figure \ref{fig:diameter} all lines are paths of length $p$.}
    \label{fig:radius}
\end{figure}
Note that the minumum eccentricity must be in a node $a_p$ since all paths from
a node $u \in G_1$ to $v \in G_2$ must pass through a node $a_p$.

We now set $p = \omega(\log n)$ sufficiently large and get the following lemma:
\begin{lemma}
    A node $a_p$ in $OV_{rad}$ has $\epsilon(a_p) = (4+o(1))p$ if $a$ is a
    hitting set and $\epsilon(a_p) = (6+o(1)p)$ if $a$ is not a hitting set.
\end{lemma}
\begin{proof}
    Since $G_2$ is a copy of $G_1$ it is enough to examine the eccentricities
    of $a_p$ in $G_1$. Assume that $a$ forms a hitting set (that is $a$ is
    not orthogonal to any $b \in B$). We then have from
    Claim~\ref{claim:notorthogonal} that $\epsilon(a_p) = 4p + O(\lg n)$.

    Assume that $a$ does not form a hitting set, since there then exists a $b
    \in B$ which is orthogonal to $a$, we have from
    Claim~\ref{claim:orthogonal} that $\epsilon(a_p) \geq 6p$.
\end{proof}
Now Theorem~\ref{thm:radius} follows from the above lemma in the same way that
Theorem~\ref{thm:diameter} followed from the discussion in the previous
section.

\subsection{Computing all eccentricities}
For Theorem~\ref{thm:eccentricities} we consider the graph from the proof of
Theorem~\ref{thm:diameter}. From the discussion in Section~\ref{sec:diameter}
we see that some node $a_r$ must have $\epsilon(a_r) = (5+o(1))p$ if there
exists an orthogonal vector pair and that all nodes $a_r$ have $\epsilon(a_r) =
(3+o(1))p$ if no such pair exists. Thus, it follows that if an algorithm can
compute a $(5/3-\delta)$ approximation of all eccentricities in $O(n^{2-\eps})$
time for any $\eps,\delta > 0$ we have a contradiction to
Conjecture~\ref{conj:ov}.

\section{Hardness of centrality problems}
In this section we present Theorems~\ref{thm:rc} and~\ref{thm:bc}.

\subsection{Reach centrality}
For reach centrality we will employ a slightly modified version
of the diameter graph from Section~\ref{sec:diameter}. The goal will be to
distinguish between a RC of $(3+o(1))p$ and $(\frac{3}{2}+o(1))p$, thus ruling
out a $(2-\delta)$-approximation for any $\delta> 0$.

Consider the $OV_{dia}$ graph constructed in the proof of
Theorem~\ref{thm:diameter}. We create a slightly modified graph, where the roots
of the $A$ and $B$-shortcut trees are connected by a path of length $2(p-\log
n)$. We denote the middle node of this path by $u$. We will now consider the
lengths of shortest paths passing through $u$. This is illustrated below in
Figure~\ref{fig:reach}

\begin{figure}[htbp]
    \centering
    \includegraphics[width=0.7\textwidth]{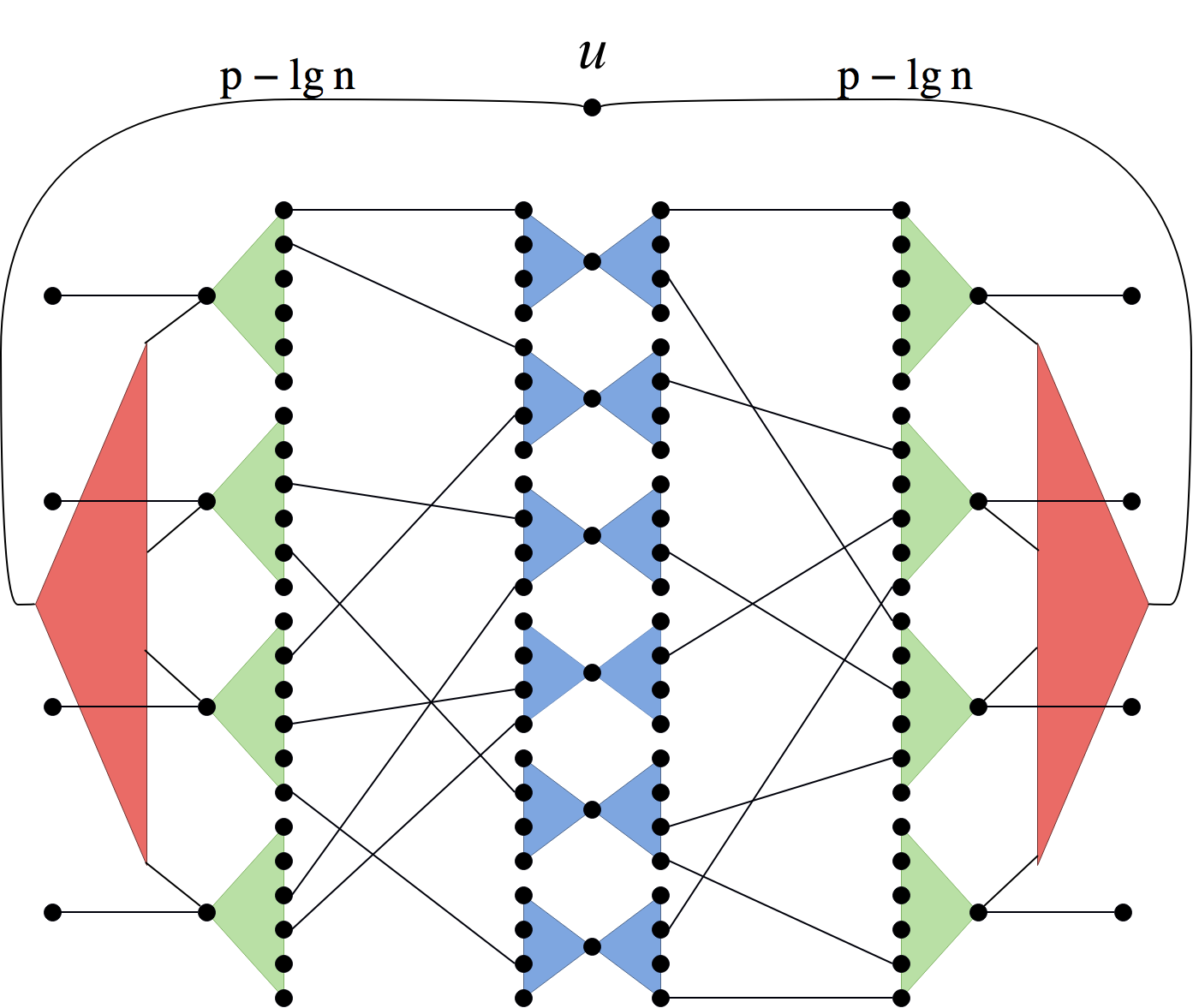}
    \caption{The modified $OV_{dia}$-graph for reach centrality}
    \label{fig:reach}
\end{figure}

\begin{claim}
    Assume that no two vectors $a\in A$ and $b\in B$, from which the modified
    $OV_{dia}$ graph was created, are orthogonal. Then $RC(u) =
    (\frac{3}{2}+o(1))p$.
\end{claim}
\begin{proof}
    First note, that the distance from $u$ to any node in one of the shortcut
    tree is at most $p$, thus for $RC(u)$ to be $(\frac{3}{2}+\delta)p$
    there must exist nodes $v,w$ with shortest path going through $u$, which
    are at distance at least $(\frac{1}{2}+\delta)p$ from the shortcut trees.
    Consider some node $v$ one the path from some $a_r$ to the $A$-shortcut
    tree at distance $(\frac{1}{2}-\delta)p$ from $a_r$ and some node $w$ on
    the path from some $b_r$ to the $B$-shortcut tree at distance
    $(\frac{1}{2}-\delta)p$ from $b_r$. Since we assumed that no two vectors of
    the OV problem are orthogonal, it follows that $d(v,w) = (3-2\delta)p
    +o(p)$ by going through some node $c_i$, and the shortest path from $v$ to
    $w$ thus cannot pass through $u$. It is easy to see that any candidate
    shortest path for $RC(u)$ has to pass through such a pair of nodes $v,w$
    the claim follows.
\end{proof}
If we now assume that there exists a pair of vectors $a\in A$ and $b\in B$ that
are orthogonal, we know that $d(a_p,b_p) = (6+o(1))p$. We will consider going
from $a_p$ to $b_p$ and show that $RC(u) = (3+o(1))p$.
\begin{claim}
    Assume that there exists $a\in A$ and $b\in B$ that are orthogonal. Then
    $RC(u) = (3+o(1))p$.
\end{claim}
\begin{proof}
    We consider the path from $a_p$ to $b_p$ going through $u$. This path has
    length $2(p + p + \log n + p - \log n) = 6p$. Observe also that any other
    path between $a_p$ and $b_p$ has length at least $6p$ by the proof of
    Theorem~\ref{thm:diameter}, thus this is a shortest path from $a_p$ to
    $b_p$ passing through $u$, and it follows that $RC(u)$ is at least $3p$. It
    is not hard to see that this is also an upper bound on $RC(u)$.
\end{proof}
It now follows that any algorithm that can compute a $(2-\delta)$-approximation
of $RC(u)$ in bounded degree graphs can distinguish between the two cases. Thus
any algorithm that can do this in time $O(n^{2-\eps})$ for any $\eps > 0$ is a
contradiction to Conjecture~\ref{conj:ov}.

\subsection{Betweenness centrality}
We will start by showing Theorem~\ref{thm:bc} for sparse graphs instead of
constant degree graphs, and then change the construction to work also in this
case. We will base our reduction on the $OV$-graph and modify it as follows:
\begin{itemize}
    \item Add two nodes $x$ and $y$, and an edge $(x, y)$.
    \item Add an edge $(a, x)$ for each $a \in A$, and edge $(b, y)$ for each
        $b \in B$
    \item Create two copies $G_1$ and $G_2$ of this graph. Denote node $x \in
        G_1$ by $x_1$ and $x \in G_2$ by $x_2$.
    \item Delete all nodes from $B$ in $G_1$
\end{itemize}

Having created the graphs $G_1$ and $G_2$ we now query the betweenness
centrality of $x_1$ in $G_1$ and of $x_2$ in $G_2$. By comparing the two values
we will be able to determine the answer to the OV problem instance.
\begin{lemma}
    $BC(x_2) > BC(x_1)$ if and only if there is an orthogonal pair $a \in A$,
    $b \in B$.
\end{lemma}
\begin{proof}
    If $BC(x_2) > BC(x_1)$ it must be the case that some shortest path $P$ with
    a node $b \in B$ as an endpoint goes through $x_2$, since these are the
    only paths not counted in $BC(x_1)$. The other endpoint of $P$ cannot be
    $y$ or any node $c \in C$ as these are at distance 1 from $b$. $P$ must
    therefore have an endpoint in $a\in A$ and since $P$ passes through $x$ it
    must be of length 3. It now follows that $a$ and $b$ form an orthogonal
    vector pair. If no such pair $a$ and $b$ exists, all such paths $P$ have
    length $2$ and do not pass through $x_2$.
\end{proof}

We are now ready to proove the full version of Theorem~\ref{thm:bc}. As for
general unweighted graphs, our reduction for betweenness centrality in constant
degree graphs, starts with the $OV$-graph, and modifies it as follows:
\begin{itemize}
    \item As in the $OV_{dia}$-graph, add vector trees at each node $v \in A
        \cup B$ and tree-pairs $c_A$, $c_B$ for each $c \in C$. Connect these
        exactly as in the $OV_{dia}$-graph.
    \item Insert a node $x$ and make it the root of a balanced binary tree with
        $n$ leaves. Connect each leaf of this \emph{$x$-tree} to a root node
        $a_r$ in a vector-tree $a \in A$ with a path of length $p$.
    \item Insert a node $y$ and a \emph{$y$-tree} connected to nodes in $B$
        exactly as the $x$-tree is with $A$.
    \item Insert a path of length $p$ between $x$ and $y$.
    \item Denote this graph by $G_1$ and create a copy $G_2$. Denote node $x \in
        G_1$ by $x_1$ and $x \in G_2$ by $x_2$.
    \item In $G_2$ create a node $b'$ for each $b \in B$ and add the edge
        $(b_r,b')$ for each such $b'$.
\end{itemize}
This construction is illustrated in Figure~\ref{fig:betweenness}.
\begin{figure}[htbp]
    \centering
    \includegraphics[width=0.8\textwidth]{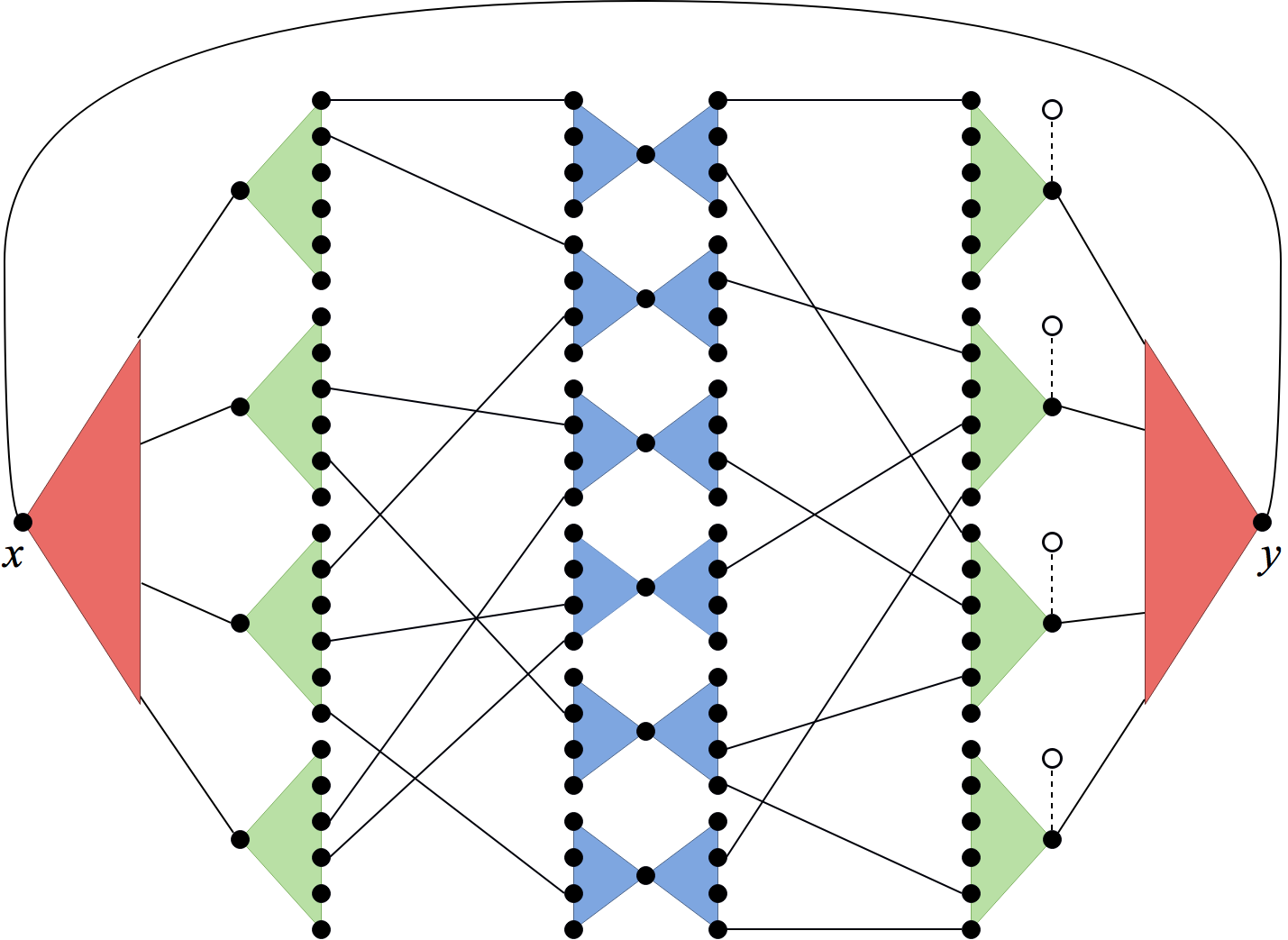}
    \caption{The graph for betweenness centrality in constant degree graphs. The white nodes along with dotted edges are the $b'$ nodes, only in $G_2$. All other lines are paths of length $p$.}
    \label{fig:betweenness}
\end{figure}

We first examine some distances in $G_1$ which will be useful.

\begin{observation}
$d(a_r, b_r) = 2p + 2 \lg n + 2 \lg d$ when $a$ and $b$ are not orthogonal
\end{observation}

\begin{observation}
$d(a_r, b_r) = 3p + 2 \lg n$ when $a$ and $b$ are orthogonal.
\end{observation}

Using these observations we can prove the following lemma.
\begin{lemma}
    $BC(x_2) > BC(x_1) + n^2 \cdot (\frac{1}{2}p + \lg d) +n\cdot (2n-2)$ if
    and only if there is an orthogonal pair $a \in A$, $b \in B$.
\end{lemma}
\begin{proof}
    We will examine the contribution to $BC(x_2)$ of each inserted node $b'$.

    Consider a vector $b \in B$ and a vector $a \in A$ that are not orthogonal.
    The shortest path from $b'$ to $a_r$ has length  $2p + 2 \lg n + 2
    \lg d$ and does not go through $x_2$. However on the path between $a_r$ and
    $x_2$ there are a number of nodes $v$ for which the shortest path between
    $b'$ and $v$ goes through $x_2$. Since $d(b', x_2) = 2p + \lg n$ and
    $d(x_2, a_r) = p + \lg n$ we get that there are $\frac{p}{2} + \lg n + \lg
    d$ such nodes $v$.

    If a vector $b \in B$ is not orthogonal to any $a \in A$ the node $b'$ will
    then contribute $\frac{p}{2} + \lg n + \lg d$ to $BC(x_2)$ for each $a$ for
    a total of $n (\frac{p}{2} + \lg n + \lg d)$ for each such vector $b$. Thus
    if there is no orthogonal pair the total contribution to $BC(x_2)$ from
    nodes $b'$ will be $n^2(\frac{p}{2} + \lg n + \lg d)$. However, by doing
    this we have overcounted the contribution of $b'$ and each node of the
    $x$-tree. Thus the actual contribution becomes $n^2\cdot(\frac{p}{2} + \lg
    d) + n\cdot (2n-2)$ as the tree has $2n-2$ nodes if we exclude the root
    $x_2$.

    If a vector $b \in B$ is orthogonal to an $a \in A$ the shortest path from
    $b'$ to $a_r$ will go through $x_2$ and likewise for all nodes $v$ on the
    path from $x_2$ to $a_r$ for a total contribution of at least $p + \lg n$
    to $BC(x_2)$

    Observe that w.l.o.g. all the shortest paths considered above and thus each
    contribute $1$ to the BC of $x_2$.
\end{proof}
Theorem~\ref{thm:bc} now follows directly from the above lemma using the same
approach as in Section~\ref{sec:diameter}.

\bibliographystyle{plain}
\bibliography{cond_lbs}

\appendix
\section{Omitted proofs}
\subsection{proofs of claims about the $OV_{dia}$ graph} \label{appendix:diaproofs}
Proof of claim \ref{claim:notorthogonal}
\begin{proof}
Any path from $a_p$ to $b_p$ must go through both of the length $p$ paths starting at the two nodes. Since $a$ and $b$ are not orthogonal, there exists a node $c$ such that there is an path of length $p$ from the vector tree at $a$ to $c_A$ and likewise from $b$'s vector tree to $c_B$, taking this path adds the height of the vector trees, the height of $c_A$ and $c_B$ and $2p$ to the total path, giving $d(a_p, b_p) \leq 4p + 2 \lg n + 2 \lg d$
\end{proof}

Proof of claim \ref{claim:orthogonal}
\begin{proof}
The distance between $a_p$ and $b_p$ must be greater when the two correspongding vectors are orthogonal, as they do not have direct paths to any node $c$. The shortest path between the two must then go through some other vector trees, w.l.o.g. assume that $a'$ and $b'$ are \emph{not} orthogonal, and $a'$ is the closest such vector tree to $a_r$, either in the shortcut tree or through some $c_A$tree. The distance between $a_ri$ and $a'_r$ through the shortcut tree is at least $2p$, and thus adds $2p$ to the distance from claim \ref{claim:notorthogonal}. To see that this is the shortest distance, note that this is the shortest possible path using the shortcut trees to get from $a$ to $a'$, and any path going through a clause would use $> 2p$ edges to get to $a'$.
\end{proof}

Proof of claim \ref{claim:vectortrees}
\begin{proof}
For any two nodes $u$ in the vector tree of a vector $a$ and $v$ in the vector tree of $a'$ the path through the $a_r$, $a'_r$ and the shortcut tree is at most $2p + 2 \lg d + 2 \lg n$.
\end{proof}

Proof of claim \ref{claim:cnode}
\begin{proof}
We consider two cases: either a node $a_p$ is the furthest from $u$ (symmetric for $b_p$), or a node $v$ in some $c'_A$ or $c'_B$ is the furthest from $u$.

For the first case any such node can be reached by traversing trees $c_A$ and $c_B$, a path of length $p$ to a vector tree rooted at $a'_r$ (possibly $a = a'$ in which case the path is shorter.), a "shortcut" of $2p + 2 \lg n$ edges through the shortcut tree to $a_r$, and $p$ edges from $a_r$ to $a_p$. For a total of $4p + 4 \lg n + \lg d$

For the second case consider the following path from $u$ to $v$, let $a$ be a vector for which there is a length $p$ path from $a$'s vector tree to $c_A$, similarly let $a'$ be such a vector for $c'_A$, the path from $u$, to $v$ through $a_r$ the shortcut tree and $a'_r$ must then have length at most $4p + 6 \lg n + 2 \lg d$ since $\lg n$ is both the height of the shortcut tree and the trees $c_A$ and $c_B$.
\end{proof}

\end{document}